%% file: main.tex
\documentclass[12pt, draftclsnofoot, onecolumn]{IEEEtran}

\usepackage{epsfig}
\usepackage{latexsym}
\usepackage{amsfonts,amsmath,color,amssymb,amsxtra, graphicx, times}
\usepackage[ansinew]{inputenc}
\usepackage{subfigure}
\usepackage{multirow}
\usepackage{algorithm,algorithmic}
\usepackage{slashbox}

\date{}
\begin{document}

%%%%%%%%%%% settings for deliverable 1 JSAC
\setlength{\textheight}{19.79205cm}
\setlength{\textwidth}{6.5in}
\setlength{\topmargin}{1.5cm}

%----------------------------------luisa's notation-------------------------
\newcommand{\edgesin}{\Gamma_{I}(v)}
\newcommand{\edgesout}{\Gamma_{O}(v)}
\newcommand{\indegree}{\delta_{I}(v)}
\newcommand{\outdegree}{\delta_{O}(v)}
\newcommand{\field}{{\mathbb F}_{q}}

\newcommand{\setKeepers}{\mathcal{N}}           % set of data keepers
\newcommand{\nrDataKeepers}{|\setKeepers|}  % number of data keepers
\newcommand{\setSrcKeepers}{\setKeepers_s}  % set of data keepers that are connected to the source
\newcommand{\nrSrcKeepers}{|\setSrcKeepers|}  % nr of data keepers that are connected to the source
\newcommand{\setTargetKeepers}{\setKeepers_r} % set of data keepers that are connected to the target
\newcommand{\nrTargetKeepers}{|\setTargetKeepers|} % nr of data keepers that are connected to the target
\newcommand{\node}{n}  % a node
\newcommand{\probB}{p_{b}}  % probability of a node being Byzantine
\newcommand{\setInformedNodes}[1]{N({#1})}  % set of informed nodes (the nodes that are in the information contact graph)
\newcommand{\setKeepersList}{D} % set of data keepers each data keeper connects to
\newcommand{\nrKeepersList}{d} % number of data keepers each data keeper connects to
\newcommand{\byzNode}[1]{I_b({#1)=1}} % node is Byzantine (indicator variable)
\newcommand{\nonByzNode}[1]{I_b({#1})=0}% node is not Byzantine
\newcommand{\contNode}[2]{{I_{c}({#1,#2})=1}} % node is contaminated at time t
\newcommand{\unContNode}[2]{{I_{c}({#1,#2})=0}} % node is not contaminated at time t
\newcommand{\setContNodes}{N_c} % set of contaminated nodes
\newcommand{\nrCont}[1]{C({#1})} % number of contaminated nodes
\newcommand{\nrUnCont}[1]{\overline{C}({#1})}   % number of uncontaminated nodes
\newcommand{\setUnContNodes}{N_{\overline{c}}} % set of uncontaminated nodes
\newcommand{\nrInfoByzantine}{N_{b}}
\newcommand{\blockProb}{\Psi}
\newcommand{\matrixP}{\mathbf{P}}
\newcommand{\markovC}{\Xi}
\newcommand{\eventA}{A}
\newcommand{\eventB}{D}
\newcommand{\setInformedNodesTrack}[1]{\setInformedNodes{#1}^{\textrm{Tracker}}}
\newcommand{\graph}[1]{G(#1)}

\newcommand{\VONEJSAC}[1]{#1}   % uncomment this for V1 JSAC deliverable
\newcommand{\VTWOJSAC}[1]{}

%----------------------------------minji's notation-------------------------
%----------------------------------Fang's notation-------------------------
\newcommand{\comment}[1]{}
\newcommand{\fig}[1]{{\itshape Figure~\ref{#1}}}
\newcommand{\alg}[1]{{\itshape Algorithm~\ref{#1}}}
\newcommand{\sect}[1]{{\itshape Section~\ref{#1}}}
\newcommand{\tabl}[1]{{\itshape Table~\ref{#1}}}
\newcommand{\eq}[1]{\eqref{#1}}
\newcommand{\secref}[1]{Section~\ref{#1}}
\newcommand{\eg}{e.g.~}
\newcommand{\theorref}[1]{{\itshape Theorem~\ref{#1}}}
\newcommand{\cororref}[1]{{\itshape Corollary~\ref{#1}}}
\newcommand{\proprref}[1]{{\itshape Proposition~\ref{#1}}}
\newcommand{\lemmarref}[1]{{\itshape Lemma~\ref{#1}}}
\newcommand{\equationrref}[1]{{\itshape Equation~\ref{#1}}}
\newcommand{\algrref}[1]{{\itshape Algorithm~\ref{#1}}}
\newcommand{\boxend}{\hfill$\blacksquare$} % box at end of definitions, theorems, etc.
\newtheorem{theorem}{Theorem}
\newtheorem{lemma}{Lemma}
\newtheorem{proposition}{Proposition}
\newtheorem{corollary}{Corollary}
\newtheorem{example}{Example}
\newtheorem{definition}{Definition}
\newtheorem{remark}{Remark}
\newcommand{\imp}[1]{\mathsf{imp}({#1}) }
\newcommand{\ie}{{\it i.e. }}
\newcommand{\etal}{\textit{et al. }}

\newcommand{\donotinclude}[1]{#1}  %%%% include figures and tables

\newcommand{\changed}[1]{{\color{blue} #1}}
\newcommand{\review}[1]{{\color{magenta} #1}}
\newcommand{\todo}[1]{{\color{red}TODO: #1}}
\newcommand{\alert}[1]{\textcolor{green}{#1}}

%----------------------------------------------------------------------------------------%
\title{\huge{On Counteracting Byzantine Attacks \\in Network Coded Peer-to-Peer Networks}}

\author{\vspace*{.5cm}MinJi Kim$^*$, Lu\'isa Lima$^*$, Fang Zhao$^*$, Jo\~{a}o Barros, Muriel M\'edard, \\Ralf Koetter, Ton Kalker, Keesook J.\ Han
\thanks{$^*$ The first three authors contributed equally to this work.}
\thanks{M. Kim, F. Zhao and M. M\'edard (\{minjikim, zhaof, medard\}@mit.edu) are with the Research Laboratory of Electronics at the Massachusetts Institute of Technology, MA USA. L. Lima (luisalima@dcc.fc.up.pt) is with the Instituto de Telecomunica\c{c}\~oes, Department of Computer Science, Faculdade de Ci\^encias, Universidade do Porto, Portugal. J. Barros (jbarros@fe.up.pt) is with the Instituto de Telecomunica\c{c}\~oes, Departamento de Engenharia Electrot\'ecnica e de Computadores, Faculdade de Engenharia da Universidade do Porto, Portugal. R. Koetter (ralf.koetter@tum.de) is with the Institute for Communications Engineering of the Technischen Universitaet Muenchen, Germany. T. Kolker (ton.kalker@hp.com) is with the Hewlett-Packard Laboratories, CA USA. K. Han (keesook.Han@rl.af.mil) is with Air Force Research Laboratory, NY USA.}
\thanks{\vspace*{.7cm}}
%\thanks{
%Part of this work was done while L. Lima was a visiting student at the Institute for Communications Engineering of the Technischen %Universitaet Muenchen and at the Research Laboratory of Electronics at the Massachusetts Institute of Technology.
%This work was supported by the European Community under grant FP7-INFSO-ICT-215252 (N-Crave Project), the Funda\c{c}\~{a}o para a Ci\^{e}ncia e Tecnologia (Portuguese Foundation for Science and Technology) under grant SFRH/BD/24718/2005, the National Science Foundation under grants ``ITR: Network Coding - From Theory to Practice" (CCR-0325496), ``XORs in the Air: Practical Wireless Network Coding" (CNS-0627021), the Air Force Office of Scientific Research (AFOSR) under grant FA9550-06-1-0155, the DARPA and the Space and Naval Warfare System Center, San Diego under Contract No. N66001-08-C-2013, and subcontract \#069145 issued by BAE Systems National Security Solutions, Inc.
%}
}
\maketitle

%----------------------------------------------------------------------------------------%
\input{abstract}

%----------------------------------------------------------------------------------------%
\newpage
\section{Introduction}\label{sect:Introduction}
\input{introduction}

\section{Background}\label{sect:background}
\input{background}

%\section{Problem setup}\label{sect:setup}
%\input{problemsetup}

\section{Impact of Byzantine attacks on P2P networks}\label{sect:impact}
\input{impact}

\section{Signature scheme for Byzantine detection}\label{sect:scheme}
\input{scheme}

\section{Overhead analysis}\label{sect:overhead}
\input{overhead}

\section{Conclusions}\label{sect:conclusions}
\input{conclusion}

%----------------------------------------------------------------------------------------%
\bibliographystyle{IEEEtran}
\bibliography{References}
\end{document}

%% file: abstract.tex
\begin{abstract}
Random linear network coding can be used in peer-to-peer networks to increase the efficiency of content distribution and distributed storage. However, these systems are particularly susceptible to Byzantine attacks. We quantify the impact of Byzantine attacks on the coded system by evaluating the probability that a receiver node fails to correctly recover a file. We show that even for a small probability of attack, the system fails with overwhelming probability. We then propose a novel signature scheme that allows packet-level Byzantine detection. This scheme allows one-hop containment of the contamination, and saves bandwidth by allowing nodes to detect and drop the contaminated packets. We compare the net cost of our signature scheme with various other Byzantine schemes, and show that when the probability of Byzantine attacks is high, our scheme is the most bandwidth efficient.
\end{abstract} 

%% file: introduction.tex
Network coding \cite{ahlswede}, an alternative to the traditional forwarding paradigm, allows algebraic mixing of packets in a network. It maximizes throughput for multicast transmissions \cite{hmmk03, ll04, lmk06}, as well as robustness against failures \cite{algebraic} and erasures \cite{reliable}. Random linear network coding (RLNC), in which nodes independently take random linear combination of the packets, is sufficient for multicast networks \cite{rlc}, and is suitable for dynamic and unstable networks, such as peer-to-peer (P2P) networks \cite{admk05, gr05}.

A P2P network is a cooperative network in which storage and bandwidth resources are shared in a distributed architecture. This is a cost-effective and scalable way to distribute content to a large number of receivers. One such architecture is the BitTorrent system \cite{bt}, which splits large files into small blocks. After a node downloads a block, it acts as a source for that particular block. The main challenges in these systems are the scheduling and management of rare blocks.

As an alternative to current strategies for these challenges, \cite{admk05, gr05} propose the use of RLNC to increase the efficiency of content distribution in a P2P solution. These schemes are completely distributed and eliminate the need of a scheduler, since each node independently forwards a random linear combination. In addition, there is a high probability that each packet a node receives is linearly independent of the previous ones, and thus, the problem of redundancy caused by the flooding approaches in traditional P2P networks is reduced. RLNC based schemes significantly reduce the downloading time and improve the robustness of the system \cite{admk05, gmr06}.

Despite their desirable properties, network coded P2P systems are particularly susceptible to {\em Byzantine attacks} \cite{perlman, liskov,Lamport} -- the injection of corrupted packets into the information flow. Since network coding relies on mixing of packets, a single corrupted packet may easily corrupt the entire information flow \cite{resilient, homomorphic1}. Furthermore, in P2P networks, there is typically no security control over the nodes that join the network and the packets that they redistribute. The topologies of the overlay graphs that arise from traditional P2P networks are often modeled as scale-free and small-world networks \cite{barabasi1999esr, PhysRevE.64.046135}, which are prone to the dissemination of epidemics, such as worms and viruses \cite{PhysRevLett.86.3200,PhysRevE.64.066112}. Several authors address these problems in coded P2P networks. We shall discuss these countermeasures in Section \ref{sect:background}. Most of these can be divided into two main categories: (i) end-to-end error correction and (ii) misbehavior detection.%, which can be carried out either packet by packet or in generation based fashion.

%In particular, RLNC has been shown to be very robust to packet losses induced by node misbehavior \cite{reliable}.

%In particular, Charles \etal \cite{elliptic} propose a signature scheme for network coded systems, based on Weil pairing on elliptic curves for pollution detection, but its computational complexity is quite high.

Motivated by these observations, we address the issues of Byzantine adversaries in coded P2P networks. The main contributions of this paper are as follows:
\begin{itemize}
\item We propose a model for the evaluation of the impact of Byzantine attacks in coded P2P networks, and provide analytical results which show that, even for a small probability of attack, the information can become contaminated with overwhelming probability.
\item We propose a new efficient, packet-based signature scheme, designed specifically for RLNC systems, to detect Byzantine attacks by checking the membership of a received packet in the valid vector space. This scheme allows an one-hop containment of the contamination.
\item We analyze the overhead in terms of bandwidth associated with our signature scheme, and compare it to that of various Byzantine detection schemes. We also show that our scheme is the most bandwidth efficient if the probability of attack is high.
\end{itemize}

This paper is organized as follows. Section \ref{sect:background} gives an overview of network coding in P2P networks and existing Byzantine detection schemes. %In Section \ref{sect:setup}, we introduce our network model.
In Section \ref{sect:impact}, we analyze the impact of Byzantine attacks on the system. We propose our signature scheme in Section \ref{sect:scheme}, and compare its overhead with other schemes in Section \ref{sect:overhead}. Finally, we conclude in Section \ref{sect:conclusions}.

%% file: background.tex
\subsection{Network coding in P2P networks}\label{sect:netcodP2P}
%Reference \cite{rlc} shows that RLNC can approach the min-cut rate for multicast networks, as the field size approaches infinity. 
References \cite{reliable, rlc} propose a \emph{random block linear network coding system} -- a simple, practical capacity-achieving code, in which every node independently constructs its linear code randomly. In such a system, a source generates information in batches of $G$ packets (called a \emph{generation}). The source then multicasts them to its destination nodes using RLNC, where only the packets from the same generation are mixed. Note that RLNC is a distributed protocol, which requires no state information; thus, making it suitable for dynamic and unstable networks where state information may change rapidly or may be hard to obtain.

Several authors have evaluated the performance of network coding in P2P networks. Gkantsidis \etal \cite{gr05} propose a scheme for content distribution of large files in which nodes make forwarding decisions solely based on local information. This scheme improves the expected file download time and the robustness of the system. Reference \cite{admk05} compares the performance of network coding with traditional coding measures in a distributed storage setting with very limited storage space with the goal of minimizing the number of storage locations a file-downloader connects to. They show that RLNC performs well without the need for a large amount of additional storage space. Dimakis {\em et al} \cite{dimakis-2007} introduce a graph-theoretic framework for P2P distributed system, and show that RLNC minimizes the required bandwidth to maintain the distributed storage architectures.

\subsection{Byzantine detection scheme for network coded systems}\label{sect:byzantimeSchemes}

\subsubsection{End-to-end error correction scheme}\label{sect:end-to-end}
Reference \cite{errorcorrection} introduces \emph{network error correction} for coded systems. They bound the maximum achievable rate in an adversarial setting, and generalize the Hamming, Gilbert-Varshamov, and Singleton bounds. Jaggi \etal \cite{resilient} introduce the first distributed polynomial-time rate-optimal network codes that work in the presence of Byzantine nodes and are information-theoretically secure. The adversarial nodes are viewed as a secondary source. The source adds redundancy to help the receivers distill out the source information from the received mixtures. This work is generalized in \cite{kschischang1, kshischang2}.

%Given an adversary who can eavesdrop on all links and jam $z$ links, their algorithm achieves a rate of $C - 2z$, where $C$ is the network capacity; given an adversary who can observe only $z_e$ links and jam $z$ links where $z_e < C - 2z$, the algorithm achieves a rate of $C - z$. These rates are the maximum achievable rate given the power of the adversary.

\subsubsection{Generation-based Byzantine detection scheme}\label{sect:generation}

Ho \etal \cite{detection} introduce an information-theoretic approach for detecting Byzantine adversaries, which only assumes that the adversary did not see all linear combinations received by the receivers. Their detection probability varies with the length of the hash, field size, and the amount of information unknown to the adversary. A polynomial hash is added to each packet in the generation. Once the destination node receives enough packets to decode a generation, it can probabilistically detect errors. The intuition behind this scheme is that if a packet is valid, then its data and hash are consistent with its coding vector; and a linear combination of valid packets is also valid.

%This generation based scheme is very cheap and sensitive. For example, with 2\% overhead ($k=50$), $\log{q} = 7$, $s = 5$, the detection probability is at least 98.9\%. Furthermore, this scheme does not require Public Key Infrastructure (PKI). However, this is a block code; therefore, will require a priori decision on the rate. In addition, the detection can only occur at a node with enough packets from a generation -- thus, can incur large delays.

\subsubsection{Packet-based Byzantine detection scheme}\label{sect:packet}
There are several signature schemes that have been presented in the literature. For instance, \cite{admk05, homomorphic2, iowa} use homomorphic hash functions to detect contaminated packets. Reference \cite{homomorphic1} suggests the use of a Secure Random Checksum (SRC) which requires less computation than the homomorphic hash function, but requires a secure channel to transmit the SRCs. In addition, \cite{elliptic} proposes a signature scheme for network coding based on Weil pairing on elliptic curves.

%Although this scheme does not require a secure channel, it is computationally expensive. Moreover, the security offered by elliptic curves that admit Weil pairing is still a topic of debate in the scientific community.

%In Section \ref{sect:scheme}, we shall propose a computationally efficient signature scheme by using the linearity property of the coded packets. This scheme does not require intermediate nodes to decode coded packets to check the validity of a packet -- thus, is efficient in delay as well.

%% file: impact.tex
In this section, we first introduce our model for evaluating the probability of a distributed denial of service attack (DDoS) caused by Byzantine nodes in a P2P network. We then present results for two distinct scenarios.

\subsection{Model}\label{sect:model}

We consider a directed graph with a set of nodes $\cal N$. A {\em source} node has a large file to be sent to \emph{receiver} nodes. The file is divided into $m$ packets. To do so, the source connects to a subset of nodes, $\setSrcKeepers\subseteq \cal N$, chosen uniformly at random, and sends each of them a different random linear combination of the original file packets. To ensure that enough degrees of freedom exist in the network, $|\setSrcKeepers|\geq m$. We refer to the nodes in $\setSrcKeepers$ as {\em level-s} nodes. A {\em tracker} node keeps track of the list of \emph{informed nodes}, $N(t)$, \emph{i.e.}, nodes that keep an information packet.

For a receiver to retrieve the file, it connects to a subset of nodes $\setTargetKeepers\subseteq \cal N$, chosen uniformly at random, with $|\setTargetKeepers|\geq|\setSrcKeepers|$. We refer to the nodes in $\setTargetKeepers$ as \emph{level-r} nodes. Note that there may be an overlap between level-s and level-r. In each time slot, one of the uninformed level-r nodes, $\node \in \setTargetKeepers\setminus \setSrcKeepers$, contacts the tracker to retrieve a random list of $d$ informed nodes, where $d<|\setSrcKeepers|$. The node  $\node$ then connects to these informed nodes through a secure overlay connection, retrieves their  packets, and stores a single random linear combination of these packets. During the same time slot, the tracker updates its list of informed nodes to $N(t)\cup \{\node\}$. This process is repeated for all nodes in $\setTargetKeepers\setminus \setSrcKeepers$, and then all level-r nodes forward their stored packets to the receiver. In order to maximize the probability of storing linearly independent combinations in level-r nodes and ensure decodability at the receiver, we set $d\ge 2$. Although we assume that each node in level-s and level-r stores only one packet, the model can be easily generalized to account for higher numbers. An example of this network model is shown in Figure \ref{network_setup}.  Note that the tracker is considered to be a trusted party in our model -- in fact, as in the case of most P2P protocols, a dishonest tracker would yield a protocol failure with overwhelming probability.

\begin{figure}[tbp]
\begin{center}
\includegraphics[scale=0.5]{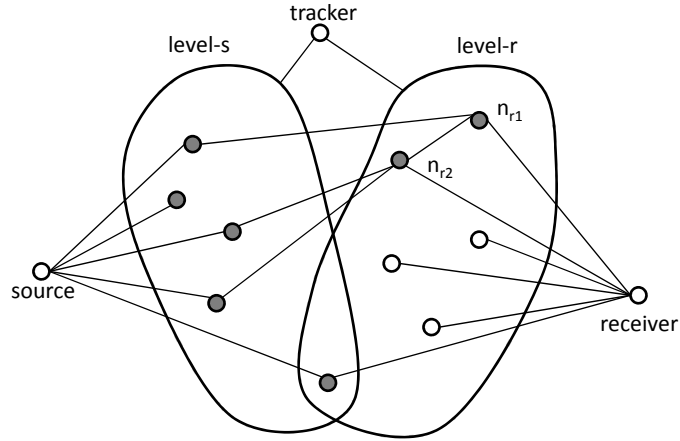}
\caption{Network model. The source is connected to the level-s nodes, and the receiver is connected to the level-r nodes. The dark nodes are the informed nodes. The level-r nodes take turns to contact the tracker, and connects to $|\setKeepersList|=2$ level-s nodes based on the list returned by the tracker. Here, nodes $n_{r1}$ and $n_{r2}$ has completed this process, and the other level-r nodes have not.}
\label{network_setup}
\end{center}
\vspace*{-0.5cm}
\end{figure}

We define an \emph{Information Contact Graph} $G(t)=\{N(t), A(t)\}$ to denote the evolving graph formed in the above process, where $N(t)$ is the list of informed nodes and $A(t)$ is the set of overlay links that connect the level-s and level-r nodes.
The probability that a node becomes a Byzantine attacker is $\probB$. An attacker corrupts the packet it stores by generating arbitrary content while complying to the standard packet format. A node independently decides whether it becomes Byzantine at the start of the file dissemination process according to $\probB$ and stays that way throughout the process. We define an indicator variable $I_b(n)$ which is 1 if node $n$ is Byzantine and 0 otherwise. The tracker has no information about which nodes are Byzantine. A \emph{contaminated packet} is a packet that is either directly corrupted by an attacker, or is a linear combination that involves at least one contaminated packet. A \emph{contaminated node} is a node that stores a contaminated packet. The \emph{blocking probability} $\Psi$ is the probability that the receiver collects at least one contaminated packet, and thus, is unable to decode the file. This is equivalent to the probability that the attacker successfully carries out a DDoS attack.

\subsection{Analysis of Impact of Byzantine Attacks}

 We now evaluate the {\em blocking probability} at the receiver. We then consider the {\em expected number of contaminated nodes} at any given time. First, we introduce necessary definitions, as follows.
 We define an indicator variable $I_c(t,n)$ which is equal to $1$ if node $\node$ is contaminated at time $t$ and $0$ otherwise.
 $\nrCont{t}$ is a random variable for the number of contaminated nodes in $\setInformedNodes{t}$, and $\nrUnCont{t} = |\setInformedNodes{t}| - \nrCont{t}$ is the number of uncontaminated nodes. The function $h(k; N, m, n)$ denotes the hypergeometric distribution, in which
 $$h(k; N, m, n) = %\mathbf{C}^m_k \mathbf{C}^{(N-m)}_{(n-k)} / \mathbf{C}^N_n
\dbinom{m}{k}\dbinom{N-m}{n-k} / \dbinom{N}{m}.$$
 Let $\nrInfoByzantine$ denote the number of informed Byzantine nodes at time $t=0$, that is, the number of Byzantine nodes in $\setSrcKeepers$. $\nrInfoByzantine$ has a binomial distribution with parameters $(\nrSrcKeepers, \probB)$.

We consider two scenarios. In \theorref{th:1}, for simplicity, we consider a static informed nodes list, in which the list kept by the tracker is fixed to $\setSrcKeepers$. In this case, level-r nodes only connect to level-s nodes. Second, in \theorref{th:2}, we generalize to the case in which the tracker updates its list of informed nodes to $\setInformedNodes{t}$, as stated in Section \ref{sect:model}.

\begin{theorem}[Static Informed Nodes List]\label{th:1}
Let $\graph{t}$ be an information contact graph in which nodes in $\setTargetKeepers$ only connect to nodes in $\setSrcKeepers$. Then its blocking probability $\blockProb$ is given by:

\vspace{-0.3cm}
\begin{equation*}
\blockProb =
1 - \sum_{y=0}^{\nrSrcKeepers} h(y; \nrDataKeepers, \nrSrcKeepers, \nrTargetKeepers) \left(\sum_{i=0}^{\nrSrcKeepers} \dbinom{\nrSrcKeepers}{i} \probB^i (1-\probB)^{\nrSrcKeepers-i}  f(i, y) \right),
\end{equation*}

where
\vspace{-0.3cm}
\begin{equation*}
f(i,y) = \left( 1- \frac{i}{\nrSrcKeepers} \right)^y \bigg{[}  (1- \probB) h(0, \nrSrcKeepers, i, \nrKeepersList) \bigg{]}^{\nrTargetKeepers - y}.
\end{equation*}

\end{theorem}

\vspace{0.2cm}

\begin{proof}
We consider two disjoint subsets of $\setTargetKeepers$:  the set of informed nodes at $t=0$, that is, $\setTargetKeepers \cap \setSrcKeepers$, and the uninformed nodes, that is, $\setTargetKeepers \backslash \setSrcKeepers$.
Let $Y$ be a random variable for the number of nodes in $\setTargetKeepers \cap \setSrcKeepers$. $Y$ has a hypergeometric distribution, $P(Y=y) = h(y; \nrDataKeepers, \nrSrcKeepers, \nrTargetKeepers)$.

We first consider $ \node \in \setTargetKeepers \cap \setSrcKeepers$. Given $\nrInfoByzantine=i$ and $Y=y$, the probability that $\node$ is uncontaminated is equal to the probability that it is not initially Byzantine, which is equal to $1 - i/\nrSrcKeepers$.
%\vspace{-0.6cm}
%\begin{equation*}
%P(\unContNode{\node}{0}| \nrInfoByzantine = i, Y=y) = 1 -  P(\byzNode{\node} | \nrInfoByzantine = i) = 1 - \frac{i}{\nrSrcKeepers}.
%\end{equation*}
Then, the probability that all nodes in $\setTargetKeepers \cap \setSrcKeepers$ are uncontaminated is:

\vspace{-0.6cm}
\begin{equation*}
P(\unContNode{\node}{0}, \forall \node \in \setTargetKeepers \cap \setSrcKeepers | \nrInfoByzantine = 1, Y = y ) = \left( 1 - \frac{i}{\nrSrcKeepers} \right)^y.
\end{equation*}

Now, at each timeslot $t>0$, a node $\node \in \setTargetKeepers \backslash \setSrcKeepers$ becomes informed. For $\node$ to be uncontaminated, it must not be Byzantine and it must connect to $\nrKeepersList$ uncontaminated nodes. Then,

\vspace{-0.6cm}
\begin{equation*}
P(\unContNode{\node}{t}| \nrInfoByzantine = i, Y=y) = (1- \probB) h(0, \nrSrcKeepers, i, \nrKeepersList).
\end{equation*}

It follows that the probability that all nodes in $\setTargetKeepers \backslash \setSrcKeepers$ are uncontaminated at time $t$ is:

\vspace{-0.7cm}
\begin{equation*}
P(\unContNode{\node}{t}, \forall \node \in \setTargetKeepers \backslash \setSrcKeepers | \nrInfoByzantine = i, Y=y) = \bigg( (1- \probB) h(0, \nrSrcKeepers, i, \nrKeepersList) \bigg)^t, \text{ for } 0 \leq t \leq \nrTargetKeepers - y.
\end{equation*}

%\vspace{-0.29cm}
Note that since $|\setTargetKeepers \backslash \setSrcKeepers|$ nodes are added, the information dissemination process ends at $t=\nrTargetKeepers-y$. Now, the probability that only uncontaminated nodes exist in $\setTargetKeepers$ at time $t = \nrTargetKeepers - y$, conditioned on $Y=y$ and $\nrInfoByzantine = i$, is:

\vspace{-0.5cm}
\begin{equation*}
f(i,y) = \left( 1- \frac{i}{\nrSrcKeepers} \right)^y \bigg{[}  (1- \probB) h(0, \nrSrcKeepers, i, \nrKeepersList) \bigg{]}^{\nrTargetKeepers - y}.
\end{equation*}

%By considering the same probability distribution for all possible $\nrInfoByzantine$, we have that the probability of only having uncontaminated nodes in $\setTargetKeepers$ at time $t = \nrTargetKeepers - y$, conditioned on $Y=y$ is

%\begin{equation}\label{equation:1}
%\sum_{i=0}^{\nrSrcKeepers} \dbinom{\nrSrcKeepers}{i} \probB^i (1-\probB)^{\nrSrcKeepers-i} \left( 1- \frac{i}{\nrSrcKeepers} \right)^y \left[  (1- \probB) h(0, \nrSrcKeepers, \nrKeepersList, i) \right]^{\nrTargetKeepers - y}
%\end{equation}

$\nrInfoByzantine$ has a binomial distribution, $Y$ has a hypergeometric distribution and they are independent of each other. Taking out these two conditions, the probability that all nodes in $\setTargetKeepers$ are uncontaminated is:

\vspace{-0.5cm}
\begin{equation*}\label{equation:2}
\gamma =
\sum_{y=0}^{\nrSrcKeepers} h(y; \nrDataKeepers, \nrSrcKeepers, \nrTargetKeepers) \left(\sum_{i=0}^{\nrSrcKeepers} \dbinom{\nrSrcKeepers}{i} \probB^i (1-\probB)^{\nrSrcKeepers-i}  f(i, y) \right).
\end{equation*}

It follows that the blocking probability is $\blockProb = 1 - \gamma$.
\end{proof}

We now consider that the list of informed nodes at the tracker is $\setInformedNodes{t}$, that is, it is updated with each new informed level-r node.

\begin{theorem}[Evolving informed nodes list]\label{th:2}

Let $\graph{t}$ be an information contact graph in which $|\setTargetKeepers \backslash \setSrcKeepers|$ are to be added to the graph by connecting to nodes in $\setInformedNodes{t}$. Then its blocking probability $\blockProb$ is:

\vspace{-0.5cm}
\begin{equation*}
\blockProb = 1 - \sum_{y=0}^{\nrSrcKeepers} h(y; \nrDataKeepers, \nrSrcKeepers,\nrTargetKeepers) \left( \sum_{i=0}^{\nrSrcKeepers} \dbinom{\nrSrcKeepers}{i} \probB^i (1-\probB)^{\nrSrcKeepers-i} f(i, y) \right),
\end{equation*}

where

\vspace{-1.1cm}
\begin{equation*}\label{equation:5}
f(i, y) = \left( 1 - \frac{i}{\nrSrcKeepers} \right)^y \Bigg{[}\prod_{t=1}^{\nrTargetKeepers-y} (1- \probB) h (0; \nrSrcKeepers +t-1, i, d)\Bigg{]}.
\end{equation*}

\vspace{0.05cm}

\end{theorem}

\begin{proof}
Recall from \theorref{th:1} that we consider two disjoint subsets of $\setTargetKeepers$, that is, $\setTargetKeepers \cap \setSrcKeepers$ and $\setTargetKeepers \backslash \setSrcKeepers$. As before, $Y$ is the number of nodes in $\setTargetKeepers \cap \setSrcKeepers$.
Again, at time $t=0$, the probability that all nodes in $\setTargetKeepers \cap \setSrcKeepers$ are uncontaminated given $\nrInfoByzantine=i$ and $Y=y$ is $(1 - i/ \nrSrcKeepers)^y$.

%\vspace{-0.6cm}
%\begin{equation*}\label{equation:2}
%P(\unContNode{\node}{0}, \forall \node \in \setTargetKeepers \cap \setSrcKeepers | \nrInfoByzantine = 1, Y = y ) = \left( 1 - \frac{i}{\nrSrcKeepers} \right)^y.
%\end{equation*}

We now consider the nodes in $\setTargetKeepers \backslash \setSrcKeepers$ and assume $\nrInfoByzantine=i, Y=y$.
At each time step, there are $\nrCont{t}$ contaminated nodes and $\nrUnCont{t} = \nrSrcKeepers+t-\nrCont{t}$ uncontaminated nodes in $\setInformedNodes{t}$.
The probability of obtaining a contaminated node at time $t+1$ is only dependent on $\nrCont{t}$ and $\nrUnCont{t}$, and thus, we can model these probabilities by Markov chains  $\markovC|{\nrInfoByzantine, Y}=\{S,\mathbf{P}\}$, in which $S$ represents the set of states and $\mathbf{P}$ represents the matrix of transition probabilities. A state in $S$ is represented by $s = (\nrCont{t},\nrUnCont{t})$. Transitions from $s$ are only possible to $s' = (\nrCont{t}+1,\nrUnCont{t})$ and to $s'' = (\nrCont{t},\nrUnCont{t}+1)$.
It is also important to note that the depth of the Markov chain is equal to $|\setTargetKeepers \backslash \setSrcKeepers | = \nrTargetKeepers - y$. The transition probabilities from $s$ when adding a node $\node$ are $P(s\rightarrow s')=P(\contNode{t+1}{\node}|\nrCont{t},\nrUnCont{t},\nrInfoByzantine,Y)$ and $P(s\rightarrow s'')=P(\unContNode{t+1}{\node}|\nrCont{t},\nrUnCont{t}, \nrInfoByzantine, Y)$. $\markovC|\nrInfoByzantine, Y$ is illustrated in ~\fig{fig:Markov} for $|\setTargetKeepers \backslash \setSrcKeepers| =2$.

%%************************************ figures ******************************************
\begin{figure}[tbp]
\centering
\includegraphics[width=12cm]{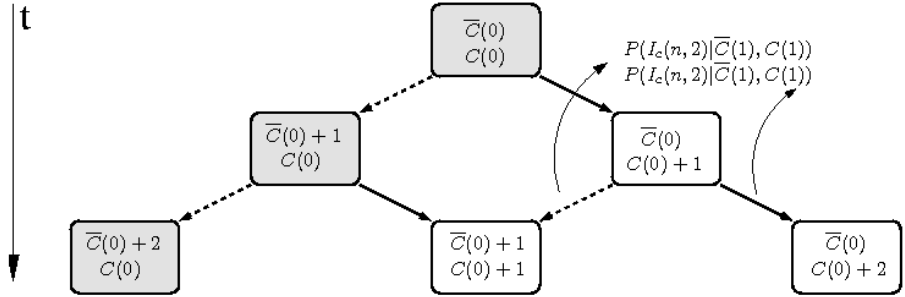}
\caption{ Markov diagram for the dissemination process, $\nrTargetKeepers-Y=2$.  The transitions to the left (dotted arrows) represent the addition of an uncontaminated node, and the transitions to the right (filled arrows) represent the addition of a contaminated node. The grey states are considered in computing $\blockProb$, that is, the states in which no contaminated nodes are added.}
\label{fig:Markov}
\vspace{-0.5cm}
\end{figure}
%%************************************ figures ******************************************

Let us denote $\nrCont{t}$ as $x$ and $t'=\nrSrcKeepers + t$, it follows that $\nrUnCont{t} = t'-x$. Now let  $p_{\{s\}}^{t}$ denote the probability of being in state $s$ at time $t$.
%We now consider the transition probabilities for $\markovC_{\nrInfoByzantine, Y}$.
% Let %$p_{(i,j)}^{t}$ denote the $ij^{\textrm{th}}$ entry of the matrix $\matrixP_{\nrInfoByzantine}^{t}$, and
% $p_{\{s\}}^{t}$ denote the entry of  $\matrixP_{\nrInfoByzantine, Y}^{t}$ corresponding to state $s$. Let us also denote $s=(\nrCont{t}, \nrUnCont{t})$ as $s=(x,y)=(x,t-x)$.
%
 $p_{\{s\}}^{t} =  p_{\{x,t-x\}}^{t}$ can be defined recursively as:
 \vspace{-0.1cm}
 \begin{equation*}
\begin{cases}
 p_{\{x,t'-x\}}^{t} = p_{\{x-1,t'-x\}}^{t-1}p_{(\{x-1,t'-x\} \rightarrow \{x,t'-x\} )}
+ p_{\{x,t'-x-1\}}^{t-1}p_{(\{x,t'-x-1\} \rightarrow \{x,t'-x\} )}, \\
 p_{(\{x,t'-x\} \rightarrow \{x+1,t'-x\} )} = 1- P(\unContNode{t}{\node}|x,t'-x, \nrInfoByzantine=i, Y=y), \\
 p_{(\{x,t'-x\} \rightarrow \{x,t'-x+1\} )} =  P(\unContNode{t}{\node}|x,t'-x,\nrInfoByzantine=i, Y=y), \\ %1 -  p_{(\{x,t-x\} \rightarrow \{x+1,t-x\} )} \\
 p_{\{i, \nrSrcKeepers-i\}}^{0}= 1.
%p_{\{1,0\}}= p_{(\{0,0\} \rightarrow \{1,0\} )} \\
%p_{\{0,1\}}= p_{(\{0,0\} \rightarrow \{0,1\} )} \\
%p_{\{1,0\}}= 1 - P_{\{0,0\}}p_{(\{0,0\} \rightarrow \{1,0\} )
%  }
\end{cases}
 \end{equation*}

%Let us define $X(n,t)$ as a random variable representing the number of uncontaminated nodes taken from a sample of $\nrKeepersList$ nodes chosen randomly out of the nodes in $\setInformedNodes{t}$. Since at time $t$, there are $\nrSrcKeepers+t$ nodes in total, out of which $\nrCont{t-1}$ are Byzantine, $X(n,t)$ follows a hypergeometric distribution, that is, $P(X(n,t) = k) = h(k; \nrSrcKeepers+t-1, \nrSrcKeepers+t-1 - \nrCont{t-1}, \nrKeepersList)$.
Now, consider that node $\node$ is active at time $t$. The probability of $\node$ being uncontaminated is the probability that it is not Byzantine and does not connect to contaminated nodes. Thus,

\vspace{-0.75cm}

\VONEJSAC{
\begin{eqnarray*}\label{equation:3}
P(\unContNode{t}{n} | \nrUnCont{t-1}, \nrCont{t-1}, \nrInfoByzantine = i, Y = y) %&=& P(\nonByzNode{n}, X(n, t) = 0, | \nrInfoByzantine = i, Y = y) \\
&=& (1-\probB) h(0; \nrSrcKeepers + t -1, \nrCont{t-1}, d).
%
%&\hspace{-5.4cm}P(\unContNode{t}{n}|\nrInfoByzantine,\nrUnCont{t-1},\nrCont{t-1}, \node \in \setTargetKeepers)  = P(\nonByzNode{\node}(t))(P(\eventA)+P(\eventB))  \\
%%&=(1-\probB)\frac{N}{n^{2}}\left(M+(n-M)h(M+t,\nrCont(t-1),L,L)\right)
%&\hspace{-0.2cm}=\frac{(1-\probB)}{\nrDataKeepers}\left(\nrSrcKeepers+(\nrDataKeepers-\nrSrcKeepers)h(\nrKeepersList; \nrSrcKeepers+t-1, \nrUnCont{t-1},\nrKeepersList)\right).%h(M+t,\nrCont(t-1),L,L)\right)
\end{eqnarray*}
}

%The probability of adding a contaminated node $\node$, $P(\contNode{t}{\node}| \nrUnCont{t-1}, \nrCont{t-1}, \nrInfoByzantine = i, Y = y)$, is simply the complement of ~\equationrref{equation:3}.
Now, notice that the probability of only having uncontaminated nodes at time $t= \nrTargetKeepers - y $ is the probability of, starting in state $(\nrCont{0}, \nrUnCont{0}) = (i, \nrSrcKeepers -i)$, ending in state $(i, \nrSrcKeepers-i+\nrTargetKeepers-y)$ after $\nrTargetKeepers- y$ steps: in that case, no contaminated node is added to the network. The probability of this event, conditioned on $\nrInfoByzantine=i$ and $Y=y$, is

\vspace{-0.7cm}
\begin{equation*}
\prod_{t=1}^{\nrTargetKeepers - y} P(\unContNode{t}{n} | \nrUnCont{t-1}, \nrCont{t-1}, \nrInfoByzantine = i, Y = y) = \prod_{t=1}^{\nrTargetKeepers-y} (1- \probB) h (0; \nrSrcKeepers +t-1, i, d).
\end{equation*}

Combining the results for sets $\setTargetKeepers \cap \setSrcKeepers$ and $\setTargetKeepers \backslash \setSrcKeepers$, we have that the probability that no contaminated nodes exist in $\setTargetKeepers$ given that $\nrInfoByzantine=i$ and $Y=y$ is given by

\vspace{-0.5cm}
\begin{equation*}\label{equation:5}
f(i, y) = \left( 1 - \frac{i}{\nrSrcKeepers} \right)^y \Bigg{[}\prod_{t=1}^{\nrTargetKeepers-y} (1- \probB) h (0; \nrSrcKeepers +t-1, i, d)\Bigg{]}.
\end{equation*}

Finally, it follows that the blocking probability at time $|\setTargetKeepers \backslash \setSrcKeepers|$ is

\vspace{-0.5cm}
\begin{equation*}
\blockProb = 1 - \sum_{y=0}^{\nrSrcKeepers} h(y; \nrDataKeepers, \nrSrcKeepers,\nrTargetKeepers) \left( \sum_{i=0}^{\nrSrcKeepers} \dbinom{\nrSrcKeepers}{i} \probB^i (1-\probB)^{\nrSrcKeepers-i} f(i, y) \right).
\end{equation*}
\end{proof}

\vspace{-0.2cm}

The results from \emph{Theorems 1} and \emph{2} are illustrated in \fig{fig:theorems}. Note that even for a small $p_b$, the blocking probability $\blockProb$ is very high. Even for the case in \theorref{th:1}, $\blockProb$ grows exponentially. This is because it is sufficient for a single level-r node to connect to a Byzantine node in level-s to contaminate the receiver. Figure \ref{fig:theorems} indicates that $\blockProb$ grows faster for the evolving informed node list than for the static informed node list). This is due to the fact that as more nodes are added to the network, the presence of contaminated nodes becomes more likely, and thus, the probability that a level-r node connects to at least one contaminated node increases. The probability $\blockProb$ also increases with other parameters such as $\nrKeepersList$, $\nrSrcKeepers$, and $\nrTargetKeepers$ since they increase the probability of level-r nodes connecting to contaminated nodes.

From the above proofs, it follows that the number of contaminated nodes in $\setInformedNodes{t}, t>0$, is dependent on the random variable $Y = |\setTargetKeepers \cap \setSrcKeepers|$. We now perform an analysis of the expected number of contaminated nodes in the network $E[\nrCont{t}]$ conditioned on $Y=y$.

First, we consider the case of the static informed nodes list, conditioned on $\nrInfoByzantine=i, Y=y$. It is clear that $E[\nrCont{0}|\nrInfoByzantine=i] = i$. Now, at each time step $t$, one contaminated node is added to $\setInformedNodes{t}$ with probability $1-P(\unContNode{\node}{t}| \nrInfoByzantine = i, Y=y)$ and thus $E[\nrCont{t}|\nrInfoByzantine=i, Y=y] = i + t(1-(1-\probB) h(0; \nrSrcKeepers, i, d))$. It follows that
\begin{equation*}
E[\nrCont{t}| Y=y] = \sum_{i=0}^{\nrSrcKeepers} \dbinom{\nrSrcKeepers}{i} \probB^i (1-\probB)^{\nrSrcKeepers-i} \bigg( i + t(1-(1-\probB) h(0; \nrSrcKeepers, i, d)) \bigg).
\end{equation*}

In the case of the evolving informed nodes list, since the states of $\markovC|\nrInfoByzantine, Y$ are representative of the number of contaminated nodes in the network, $E[\nrCont{t}|\nrInfoByzantine, Y]$ has a direct correspondence to the expected state the Markov Chain is in after $t$ time steps; therefore:

\vspace{-0.5cm}
\begin{equation*}
E[\nrCont{t}| Y=y] = \sum_{i=0}^{\nrSrcKeepers} \dbinom{\nrSrcKeepers}{i} \probB^i (1-\probB)^{\nrSrcKeepers-i} \bigg(  \sum_{x=i}^{i+t}  x  p^t_{\{x, t'-x\}}\bigg).
\end{equation*}

In order to visualize these results, we take the expected value of $Y$ for the set of parameters chosen in \fig{fig:theorems}, which is equal to 1. Then, we plot $E[\nrCont{t}| Y=1]$ for the static and evolving informed node lists. It is shown in \fig{fig:exp} that the expected number of contaminated nodes in the static case is linear with time. For small probabilities $\probB$, the $E[\nrCont{t}| Y=1]$ is higher for the evolving case; as $\probB$ increases, the values for both cases become similar.

%%*****************************figures*****************************%

\begin{figure}[tbp]
\centering
\includegraphics[width=11cm]{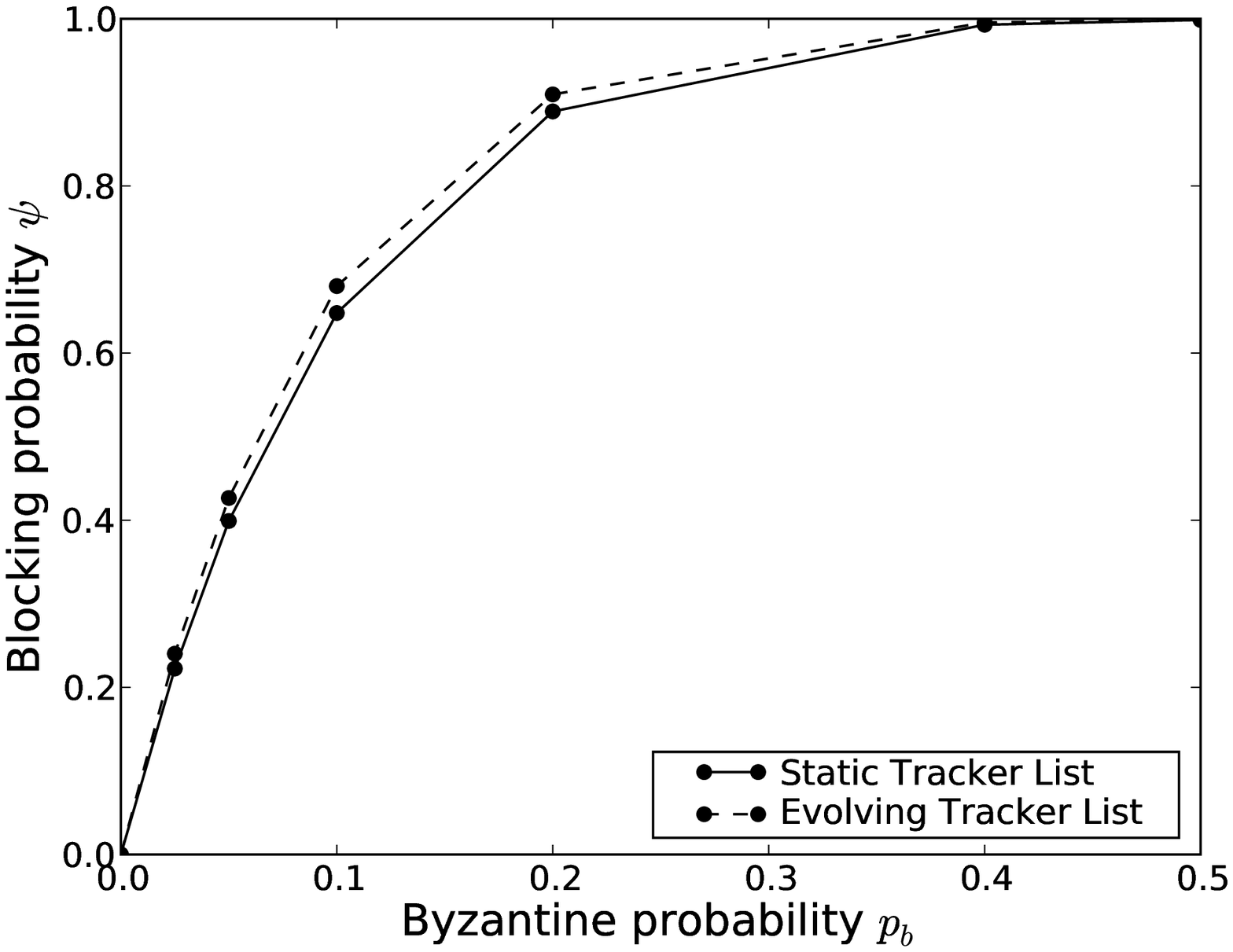}
\caption{Blocking probability in function of $\probB$ for $\nrDataKeepers=30$, $\nrSrcKeepers=5$, $\nrTargetKeepers=6$ and $\nrKeepersList=3$. The results for the static and evolving informed nodes list are shown in full and dashed, respectively.}
\label{fig:theorems}
\vspace{-0.4cm}
\end{figure}

\begin{figure}[tbp]
\centering
\includegraphics[width=11cm]{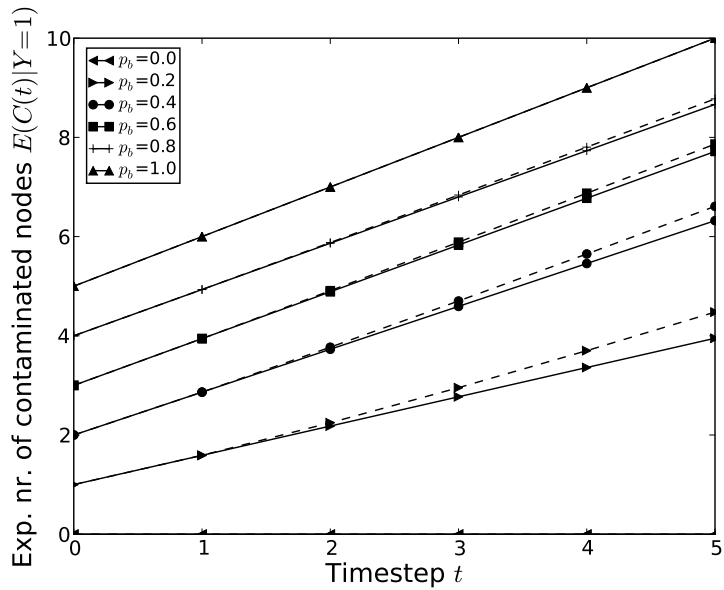}
\caption{Expected number of contaminated nodes in function of time, for $\nrDataKeepers=30$, $\nrSrcKeepers=5$, $\nrTargetKeepers=6$, $\nrKeepersList=3$ and $Y=1$. The results for the static and evolving informed nodes list are shown in full and dashed, respectively.}
\label{fig:exp}
\vspace{-0.4cm}
\end{figure}

%%*****************************figure*****************************%

%% file: scheme.tex
From the previous Section, we can see that coded P2P networks are highly vulnerable to Byzantine attacks, and the contamination can quickly spread throughout the network. Although we only consider a particular network model in Section \ref{sect:impact} for the purpose of analysis, such problems exist in all network coded systems. Therefore, it is desirable to have a signature scheme that checks the validity of each received packet without decoding the whole file. Then the contamination can be contained in one-hop, and we can avoid the decoding delay. In uncoded systems, the source knows all the packets being transmitted in the network, and therefore, can sign each one of them. However, in a coded system, each node produces ``new" packets, and standard digital signature schemes do not apply. Previous work that attempts to solve this problem is based on homomorphic hash functions \cite{admk05, homomorphic2, iowa}, Secure Random Checkup \cite{homomorphic1}, or Weil pairing on elliptic curves \cite{elliptic}. In this section, we introduce a novel signature scheme for the coded system based on the Discrete Logarithm problem.

We consider a directed graph with a set of nodes $\cal N$. A {\em source} node has a large file to be sent to \emph{receiver} nodes. The file is divided into $m$ packets. A node in the network receives linear combinations of the packets from the source or from other nodes. In this framework, a node is also a server to packets it has downloaded, and always sends out random linear combinations of all the packets it has obtained so far to other nodes. When a receiver has received $m$ linearly independent packets, it can re-construct the whole file.
We denote the $m$ original packets as $\mathbf{\bar v}_1, ...,\mathbf{\bar v}_m$, and view them as elements in $l$-dimensional vector space $\mathbb{F}_p^l$, where $p$ is a prime. The source node adds coding vectors to create  ${\mathbf v}_1, ..., {\mathbf v}_m$, ${\mathbf v}_i = (0, ..., 1, ..., 0, {\bar v}_{i1}, ...,{\bar v}_{il})$,
where the first $m$ elements are zero except the $i$th element which is 1, and $\bar v_{ij}\in \mathbb{F}_p$ is the $j$th element in $\mathbf{\bar v}_i$. A packet ${\mathbf w}$ received by a node is a linear combination of these vectors,
\[
{\mathbf w} = \sum_{i=1}^{m}\beta_i{\mathbf v}_i,
\]
where $(\beta_1, ..., \beta_m)$ is the global coding vector.

The key observation for our signature scheme is that the vectors ${\mathbf v}_1, ..., {\mathbf v}_m$ span a subspace $V$ of ${\mathbb F}_p^{m+l}$, and a received vector ${\mathbf w}$ is a valid linear combination of vectors ${\mathbf v}_1, ..., {\mathbf v}_m$ if and only if it belongs to $V$. Our scheme is based on standard modulo arithmetic (in particular the hardness of the Discrete Logarithm problem) and on an invariant signature for the linear span $V$. Each node verifies the integrity of a received vector ${\mathbf w}$ by checking the membership of ${\mathbf w}$ in $V$ based on the signature.

Our signature scheme is defined by the following ingredients:
\begin{itemize}
\item $q$: a large prime number such that $p$ is a divisor of $q-1$. Note that standard techniques, such as that used in Digital Signature Algorithm (DSA) \cite{dsa}, apply to find such $q$.
\item $g$: a generator of the group $G$ of order $p$ in ${\mathbb F}_q$. Since the order of the multiplicative group ${\mathbb F}_q^*$ is $q-1$ (a multiple of $p$), we can always find a subgroup, $G$, with order $p$ in ${\mathbb F}_q^*$.
\item Private key: ${\mathbf K}_s = \{\alpha_i\}_{i=1,...,m+l}$, a random set of elements in ${\mathbb F}^*_p$, only known to the source.
\item Public key: ${\mathbf K}_p = \{h_i=g^{\alpha_i}\}_{i=1, ..., m+l}$, signed by some standard signature scheme, e.g., DSA, and published by the source.
\end{itemize}

\noindent To distribute a file in a secure manner, the signature scheme works as follows.
\begin{enumerate}
\item Using the vectors ${\bf v}_1, ...,{\bf v}_m$ from the file, the source finds a vector ${\mathbf u}=(u_1, ..., u_{m+l})\in{\mathbb F}_p^{m+l}$ orthogonal to all vectors in $V$. Specifically, the source finds a non-zero solution, ${\bf u}$, to the set of equations ${\bf v}_i\cdot {\bf u}=0$ for $i=1,...,m$.
\item The source computes the vector ${\mathbf x}=(u_1/\alpha_1, u_2/\alpha_2, ..., u_{m+l}/\alpha_{m+l})$.
\item The source signs ${\mathbf x}$ with some standard signature scheme and publishes ${\bf x}$. We refer to the vector ${\bf x}$ as the signature of the file being distributed.
\item The client node verifies that ${\bf x}$ is signed by the source.
\item When a node receives a vector ${\mathbf w}$ and wants to verify that ${\mathbf w}$ is in $V$, it computes
\[
d=\prod _{i=1}^{m+l}h_i^{x_iw_i},
\]
and verifies that $d=1$.
\end{enumerate}

To see that $d$ is equal to 1 for any valid ${\mathbf w}$, we have
\[
d = \prod _{i=1}^{m+l}h_i^{x_iw_i}
  = \prod_{i=1}^{m+l}(g^{\alpha_i})^{u_iw_i/\alpha_i}
  = \prod_{i=1}^{m+l}g^{u_iw_i}
  = g^{\sum_{i=1}^{m+l}(u_iw_i)} = 1,
\]
where the last equality comes from the fact that ${\mathbf u}$ is orthogonal to all vectors in $V$.

Next, we show that the system described above is secure. In essence, the theorem below shows that given a set of vectors that satisfy the signature verification criterion, it is provably as hard as the Discrete Logarithm problem to find new vectors that also satisfy the verification criterion other than those that are in the linear span of the vectors already known.

\begin{definition}
Let $p$ be a prime number and $G$ be a multiplicative cyclic group of order $p$. Let $k$ and $n$ be two integers such that $k<n$, and ${\bf \Gamma}=\{h_1, ...,h_n\}$ be a set of generators of $G$. Given a linear subspace, $V$, of rank $k$ in ${\mathbb F}_p^n$ such that for every ${\mathbf v}\in V$, the equality ${\bf \Gamma} ^{\mathbf v}\triangleq \prod _{i=1}^{n}h_i^{v_i} = 1$ holds, we define the $(p,k,n)$-Diffie-Hellman problem as the problem of finding a vector ${\mathbf w}\in {\mathbb F}_p^n$ with ${\bf \Gamma}^{\mathbf w}=1$ but ${\mathbf w}\notin V$.
\end{definition}

By this definition, the problem of finding an invalid vector that satisfies our signature verification criterion is a $(p,m,m+l)$-Diffie-Hellman problem. Note that in general, the $(p,n-1,n)$-Diffie-Hellman problem has no solution. This is because if $V$ has rank $n-1$ and a ${\bf w'}$ exists such that ${\bf \Gamma}^{\mathbf w'}=1$ and ${\mathbf w'}\notin V$, then $w'+V$ spans the whole space, and any vector ${\bf w}\in {\mathbb F}_p^n$ would satisfy ${\bf \Gamma}^{\mathbf w}=1$. This is clearly not true, therefore, no such ${\bf w'}$ exists.

\begin{theorem}
For any $k<n-1$, the $(p,k,n)$-Diffie-Hellman problem is as hard as the Discrete Logarithm problem.
\end{theorem}

\begin{proof} Assume there exists an efficient algorithm to solve the $(p,k,n)$-Diffie-Hellman problem, and we wish to compute the discrete logarithm $\log_g(z)$ for some $z=g^x$, where $g$ is a generator of a cyclic group $G$ with order $p$. We can choose two random vectors ${\mathbf r}=(r_1, ..., r_n)$ and ${\mathbf s}=(s_1,...,s_n)$ in ${\mathbb F}_p^n$, and construct ${\bf \Gamma} = \{h_1,...,h_n\}$, where $h_i=z^{r_i}g^{s_i}$ for $i=1,...,n$. We then find $k$ linearly independent (and otherwise random) solutions ${\mathbf v}_1,...,{\mathbf v}_k$ to the equations
\[
{\mathbf v}\cdot {\mathbf r}=0 \; \; \mbox{and} \; \; {\mathbf v}\cdot {\mathbf s}=0.
\]
Note that there exist $n-2$ linearly independent vector solutions to the above equations. Let $V$ be the linear span of $\{{\mathbf v}_1,...,{\mathbf v}_k\}$, then any vector ${\mathbf v}\in V$ satisfies ${\bf \Gamma} ^{\mathbf v}=1$. Now, if we have an algorithm for the $(p,k,n)$-Diffie-Hellman problem, we can find a vector ${\mathbf w}\notin V$ such that ${\bf \Gamma}^{\mathbf w}=1$. This vector would satisfy ${\mathbf w}\cdot (x{\mathbf r}+{\mathbf s})=0$. Since ${\mathbf r}$ is statistically independent from $(x{\mathbf r}+{\mathbf s})$, with probability greater than $1-1/p$, we have ${\mathbf w}\cdot {\mathbf r}\ne 0$. In this case, we can compute
\[
\log_g(z)=x=\frac{{\mathbf w}\cdot {\mathbf s}}{{\mathbf w}\cdot {\mathbf r}}.
\]
This means the ability to solve the $(p,k,n)$-Diffie-Hellman problem implies the ability to solve the Discrete Logarithm problem.
\end{proof}

This proof is an adaptation of a proof in an earlier publication by Boneh {\em et. al} \cite{bf99}.

Our signature scheme makes use of the linearity property of RLNC, and enables the nodes to check the integrity of packets without a secure channel, unlike the homomorphic hash function or SRC schemes \cite{homomorphic1, homomorphic2}. In addition, our scheme does not require the nodes to decode coded packets to check their validity -- thus, is efficient in terms of delay. The computation involved in the signature generation and verification processes is very simple. Furthermore, our scheme uses the Discrete Logarithm problem, which is more standardized and widely used, compared to the recently developed Weil pairing problem used in \cite{elliptic}. Lastly, we note that our signature scheme is rateless, which is not the case in end-to-end or generation based detection schemes. 

%% file: overhead.tex
%In this section, we analyze the performance of our signature scheme and other existing schemes in Sections \ref{sect:performance} and \ref{sect:overheadOthers}, respectively. In Section \ref{sect:comparison}, we compare the performance of the various schemes under varying probability of attack.

%For the packet-based Byzantine detection scheme from Section \ref{sec:bg-packet}, $v$ is given the public key $\mathbf{K_p}$; for the generation-based Byzantine detection scheme from Section \ref{sec:bg-generation}, $v$ is allowed to decode a generation, if all the packets of the given generation goes through it.

%In the remaining of the paper, we shall focus on a single non-malicious node $n$, and the overhead associated with the detection schemes through this one node. This is a reasonable approach, since we are not concerned with how a malicious node uses its bandwidth.

In the previous Sections, we showed that our signature scheme is beneficial, as even a small amount of attack can have a devastating effect in coded networks. However, we have not shown that this scheme is efficient in terms of bandwidth (\ie overhead of augmenting the signature scheme), and indeed, it is not always the case that our signature scheme is desirable. We now study the cost and benefit of the following three Byzantine schemes: 1) our signature scheme proposed in Section \ref{sect:scheme}, 2) end-to-end error correction scheme \cite{resilient}, and 3) generation-based Byzantine detection scheme \cite{reliable}. If we implement Byzantine detection schemes, we can detect contaminated data, drop them, and therefore, only transmit valid data; however, this benefit comes with the overhead of the schemes in the forms of hashes and signatures. It is important to note that, for the dropped data, the receivers perform erasure correction, which is computationally lighter than error correction; thus, there is no need of retransmissions.

We consider a node $\node \in \mathcal{N}$ in the network as in Section \ref{sect:scheme}. Node $\node$ wishes to check the validity of the data it forwards. Assume that node $\node$ receives $M$ packets per time slot. Recall that $m$ is the number of packets in a file and $l$ is the length of each packet, therefore, each packet consists of $(m+l)$ symbols. If $\node$ detects an error, then it discards that data; otherwise, it forwards the data. The probability that $\node$ receives a contaminated packet is $p_n$ as shown in Figure \ref{fig:trustednode}. Note that the probability $p_n$ of an attack is topology dependent. However, in order to compare the performance of various schemes, we use a generic per node model to examine the overhead incurred at a node. We assume that there is an external model of vulnerability which gives an estimate of $p_n$. Note that the blocking probability $\Psi$ analyzed in Section \ref{sect:impact} provides such an estimate.

\begin{figure}[tbp]
\centering
\hspace*{-.7cm}
\includegraphics[width=0.7\textwidth]{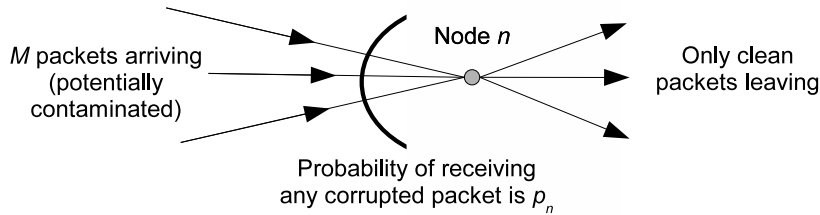} \caption{Diagram of a node $\node$ in a network}\label{fig:trustednode}
\vspace*{-.5cm}
\end{figure}

\subsection{Overhead analysis of our packet-based signature scheme}\label{sect:performance-Ours}

We examine the overhead incurred by our signature scheme. Recall from Section \ref{sect:scheme}, the file size is $ml\log{p}$ bits. The file is divided into $m$ packets, each of which is a vector in $\mathbb{F}_p^l$. Thus, the overhead of the RLNC scheme is $m/l$ times the file size, and in practical networks $m\ll l$.

The initial setup of our signature scheme involves the publishing of the public key, ${\bf K}_{p}$, which is $(m+l)\log(q)$ bits. In typical cryptographic applications, the sizes of $p$ and $q$ are 20 bytes (160 bits) and 128 bytes (1024 bits), respectively; thus, the size of ${\bf K}_{p}$ is approximately $6(m+l)/ml$ times the file size. This overhead is negligible as long as $6 \ll m\ll n$. For example, if we have a file of size 10MB, divided into $m=100$ packets, then the overhead is approximately 6\%. We note that the public key ${\bf K}_p$ cannot be fully reused for multiple files, as it is possible for a malicious node to generate a vector which is not a valid linear combination of the original vectors yet satisfies the check $d=1$ using information obtained from previously downloaded files. We do not provide the details of this for want of space.

To prevent this from happening, we can redistribute keys for each additional file in one of the two methods below. The first method consists of publishing a new public key ${\bf K}_p$ for each file, which would incur an overhead of $6(m+l)/ml$ times the file size. Note that if we republish ${\bf K}_p$ for every file, we can reuse the signature ${\bf x}$. The second method is to update  ${\bf K}_p$ partially and generate a new ${\bf x}$ for each file. This incurs less overhead than the previous method, however, requires a high variability in ${\bf w}$ for it to be secure. This update incurs negligible amount of overhead as well. For example, for a 10MB file, the overhead is less than 0.1\%.

The initial ${\bf K}_p$ distribution costs approximately 6\% of our file size, and the incremental update of ${\bf K}_p$ and ${\bf x}$ is much less than 6\% if we use the second method. Therefore, we shall denote the overhead associated with our signature by $o_p \triangleq \frac{6}{100}(m+l)$ symbols per packet, \ie 6\% overhead.

If $n$ detects an error in a packet, then it discards it -- by doing so, $n$ can filter out all the contaminated packets and use its bandwidth to transmit only valid packets. Therefore, $n$ only forwards on average $1 - \frac{o_p}{m+l}$ fraction of the data received.

%destination nodes perform erasure correction on the packets that have been dropped, which is computationally cheaper than error correction required for the error correction scheme in Section \ref{sec:nodetection}. However, each packet needs to contain enough information (a polynomial hash) that can be used to verify its integrity. This overhead associated with the Byzantine detection scheme reduces the rate at which data is transmitted.

%The overhead associated with our signature scheme can be analyzed as follows. By discarding the contaminated packets, node $v$ can on average save its bandwidth by $M(m+l)p_n$ symbols per time slot at a cost of $o_p$ symbols per packet. Therefore, the ratio between the overhead and the total data received is:
Our signature scheme costs $o_p M$ symbols per time slot. However, by discarding the contaminated packets, node $\node$ can on average save its bandwidth by $M(m+l)p_n$ symbols per time slot. Therefore, the net cost of the signature scheme as a fraction of the total data received is:
\begin{equation}\label{eq:packet}
\frac{\max\{0, Mo_p - M(m+l)p_n)\}}{M(m+l)} = \frac{\max\{0, o_p - (m+l)p_n\}}{m+l}.
\end{equation}

When $p_n$ is high, then checking each packet for error saves on bandwidth -- \ie $(o_p - (m+l)p_n) < 0$, which shows that the cost of the signature scheme is canceled by the bandwidth gained from dropping the corrupted packets. Therefore, this approach is the most sensible when the network is unreliable or under heavy attack.

\subsection{Overhead analysis of end-to-end error correction}\label{sect:performance-E2E}

In this subsection, we shall use the rate-optimal error correction codes from Jaggi \etal \cite{resilient}. As long as the attack is within the network capacity, this scheme allows the intermediate nodes to transmit at the remaining network capacity, \ie the end-to-end network capacity minus the capacity the adversary can contaminate. In this scenario, node $n$ just naively performs RLNC and forwards the data it has received. Therefore, node $n$ transmits on average $M(m+l)p_n$ contaminated symbols. Thus, the net cost as a fraction of the total data received is:
\begin{equation}\label{eq:errorcorrection}
\frac{M(m+l)p_n}{M(m+l)} = p_n.
\end{equation}

%and the destination nodes employ error correction to retrieve reliable packets.
%It is important to note that error correction is computationally more expensive than erasure correction, which shall be used when Byzantine detection schemes are used.

\subsection{Overhead analysis of generation-based Byzantine detection scheme}\label{sect:performance-generation}

We now analyze the performance of the algorithm proposed by Ho \etal \cite{detection}, which uses random block linear network coding with generation size $G$ (although we have focused on RLNC so far, it is possible to extend these results by considering $m$ as the generation size $G$). This scheme is very cheap -- with 2\% overhead, the detection probability is at least 98.9\%. We denote the overhead associated with this scheme by $o_g \triangleq \frac{2}{100}(m+l)G$ symbols per generation.

After collecting enough packets from the generation, node $n$ checks for possible error in the generation, which can incur large delay. If $n$ detects an error, it discards the entire generation of $G$ packets; otherwise, it forwards the data. This scheme requires only one hash for the entire generation -- saving bits on the hashes compared to our signature scheme. However, it can be inefficient, as one contaminated packet can cause $n$ to discard an entire generation.
%The destination nodes perform erasure correction on the generations that have been dropped, which is computationally cheaper than error correction required in Section \ref{sec:nodetection}.

The probability $p_g$ of dropping a generation of $G$ packets is given by:
\[p_g = 1 - \Pr(\text{All $G$ packets are valid})= 1- (1-p_n)^G.\]

%Therefore, the probability that a generation is forwarded by $v$ is $1 - p_g = (1-p)^G$; node $v$ is expected to transmit $(1-p)^G nG$ bits per unit time. By similar analysis as in Section \ref{sec:packets}, the fraction of actual data bits of the $(1-p)^G nG$ transmitted bits is $1 - \frac{h_g}{nG}$.

The cost and benefit of this scheme includes three components: (i) the hash of $o_g$ symbols per generation, (ii) valid packets which are discarded if the generation is deemed contaminated, and (iii) bandwidth saved by dropping contaminated packets. The expected number of valid symbols dropped per generation is $p_g(1-p_n)(m+l)G$. The expected number of contaminated symbols per generation is $p_n (m+l)G$. Thus, the net cost as a fraction of the total data received is:
\begin{equation}\label{eq:generation}
\frac{\max\{0, o_g + p_g(1-p_n)(m+l)G - p_n(m+l)G\}}{(m+l)G}.
\end{equation}

For this scheme to work, $n$ needs to receive at least $G$ packets from each generation to decode and detect errors. This may seem to indicate that this scheme is only applicable as an end-to-end scheme, but it can be extended to a \emph{local} Byzantine detection scheme as shown in Figure \ref{fig:generationsplit}.

\begin{figure}[tbp]
\begin{center}\hspace*{-.5cm}
\includegraphics[width=0.75\textwidth]{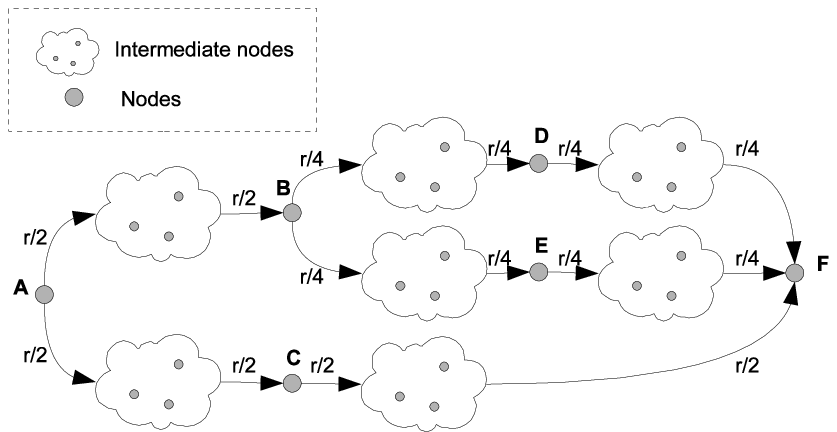}\end{center} \vspace*{-.5cm}\caption{Network with non-malicious nodes $\mathbf{A}$, $\mathbf{B}$, $\mathbf{C}$, $\mathbf{D}$, $\mathbf{E}$, and $\mathbf{F}$ where node $\mathbf{A}$ is transmitting at a total rate of $\mathbf{r}$ to node $\mathbf{F}$; however, $\mathbf{A}$ sends half of its data through $\mathbf{B}$ and the other half through $\mathbf{C}$. Therefore, $\mathbf{B}$ and $\mathbf{C}$ can check the validity of the \emph{sub-generation} they receive, where by sub-generation, we mean a collection of $G/2$ encoded packets from $\mathbf{A}$. By a similar argument, $\mathbf{D}$, $\mathbf{E}$, and $\mathbf{F}$ can check the validity of a sub-generation of $G/4$, $G/4$, and $G$ packets from $\mathbf{A}$, respectively. }
\label{fig:generationsplit}
\vspace*{-.5cm}
\end{figure}

%\subsubsection{Generation size $G$ in the generation-based scheme}\label{sect:generationsize}

%In Figure \ref{fig:generations}, given $p_n$, as $G$ increases, t
The cost of the generation-based scheme increases dramatically with $G$. If $G$ is large enough, the probability of at least one corrupted packet in a generation is high even for small $p_n$. Thus, a large $G$ is undesirable, as almost every generation is found faulty and dropped, making the throughput go to zero. This can be verified with an asymptotic analysis of Equation \ref{eq:generation}:
\begin{equation*}
\lim_{G\rightarrow \infty} \frac{\max\{0, o_g + p_g(1-p_n)(m+l)G - p_n(m+l)G\}}{(m+l)G} \rightarrow \max\{0,1-2p_n\}.
\end{equation*}

Note in Figure \ref{fig:generations} that the cost peaks at $p_n \approx 0.2$. At $p_n \approx 0.2$, the scheme drops many generations for a few corrupted packets. Thus, at a moderate rate of attack, the generation-based scheme suffers. When $p_n < 0.2$, the generation-based scheme does well, since $p_n$ is low and the cost of hash is distributed across $G$ packets. As $p_n$ increases to 0.5 from 0.2, the throughput to the receiver decreases as more generations are dropped. When $p > 0.5$, this scheme discards almost all generations, thus, the expected throughput is near zero.

\begin{figure}[tbp]
\begin{center}\hspace*{-.5cm}
\includegraphics[width=0.7\textwidth]{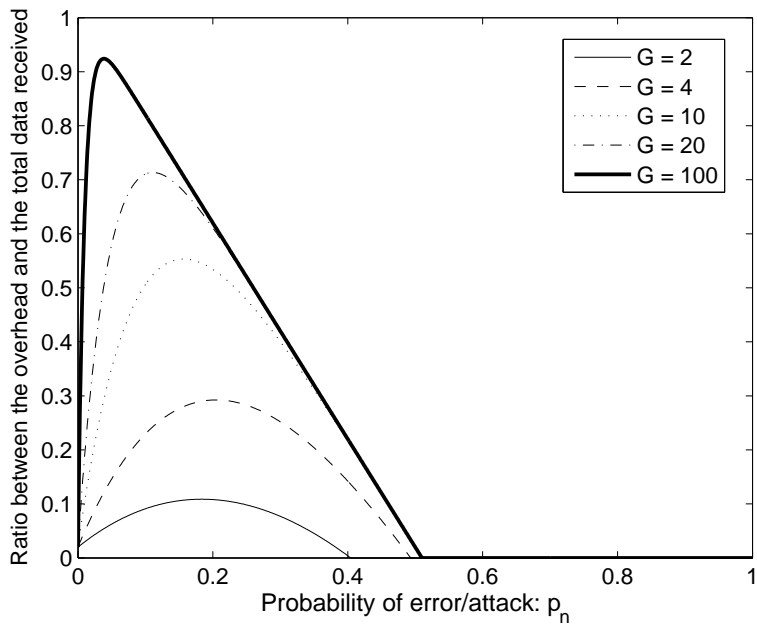}\end{center}\vspace*{-.5cm} \caption{Ratio between the expected overhead and the total data received by a node for generation-based detection with generation size $G$, packet size $1000$ bits, and hash size $o_g = \frac{2}{100}(m+l)G$ symbols per generation}\label{fig:generations}
\vspace*{-.5cm}
\end{figure}

\subsection{Trade-offs and comparisons}\label{sect:comparison}

%\subsubsection{A comparison of coded and non-coded systems}\label{sect:nocoding}

%A non-coded system, unlike its coded counterparts, requires state information. As a result, the most effective attack on a routing network is an attack on the control traffic. However, as we noted in Sections \ref{sect:Introduction} and \ref{sect:netcodP2P}, network coded systems are robust against such an attack. In addition, for an effective Byzantine detection in a routing network, all nodes in the network needs to be authenticated; therefore, each packet carries a signature and a hash to verify the identity of the source and the content of the packet. Thus, a non-coded system would incur overhead similar to that of packet-based scheme without the throughput gain due to network coding. There are various literature on the overhead analysis of secure routing protocols, especially for wireless ad hoc networks \cite{hubaux}\cite{marti}. In \cite{marti}, it has been shown that these routing protocols can incur up to 24\% overhead; making the cost of detection non-negligible without the performance benefits of coding.

%\subsubsection{A comparison of the three Byzantine detection schemes}\label{sect:threecomparison}

In Figures \ref{fig:comparison} and \ref{fig:blownup}, we compare the three schemes. As mentioned in Section \ref{sect:performance-E2E}, the expected cost of error correction scheme is linearly proportional to $p_n$. Therefore, for large $p_n$, this scheme performs badly. However, this simple scheme where a node naively forwards all data it receives outperforms the detection schemes when $p_n$ is low ($p_n <0.03$). When $p_n$ is small, the overhead of detection exceeds the cost introduced by the attackers.

When $p_n$ is low, the overhead of our signature is costly, since we are devoting $o_p$ symbols per packet to detect an unlikely attack. In such a setting, the generation-based scheme performs well, as it distributes the cost of the hash ($o_g$ symbols) over $G$ packets. However, as $p_n$ increases, the cost of our signature becomes negligible since the bandwidth wasted by contaminated packets increases; thus, our signature scheme outperforms the generation-based scheme. However, it is important to note that we underestimate the overhead associated with our signature scheme in this paper as we do not take into account the public key distribution cost, which the generation-based scheme does not require. Thus, depending on the public key distribution infrastructure used and the frequency of key renewal, our scheme will incur a higher overhead -- resulting in an outward shift in the overhead in Figure \ref{fig:comparison}.

We briefly note the computational cost of implementing these schemes. When using our signature scheme or the generation-based detection scheme, node $n$ does not waste its bandwidth in transmitting contaminated data by dropping a single packet or an entire generation. Furthermore, there is no need of retransmission of the dropped data as the receivers can perform erasure correction on the packets or the generations that have been dropped. It is important to note that for the end-to-end error correction scheme, the receivers need to perform error correction, which is computationally more expensive than erasure correction.

\begin{figure}[tbp]
\begin{center}\hspace*{-.5cm}
\includegraphics[width=0.8\textwidth]{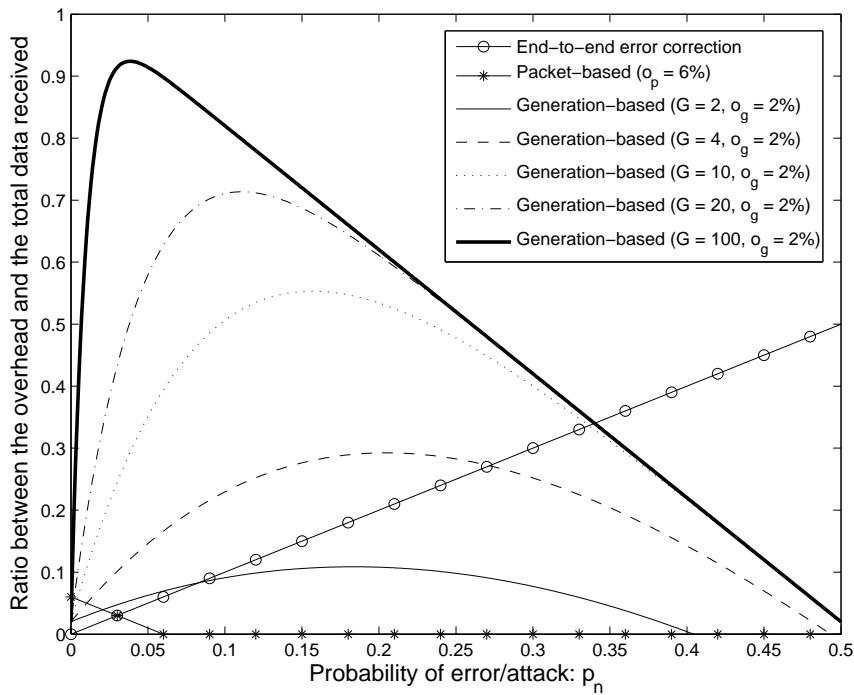}\end{center}\vspace*{-.5cm} \caption{Ratio between the expected overhead and the total data received by a node with $o_p = \frac{6}{100}(m+l)$, $o_g= \frac{2}{100}(m+l)G$}
\label{fig:comparison}
\vspace*{-.6cm}
\end{figure}

\begin{figure}[tbp]
\begin{center}\hspace*{-.5cm}
\includegraphics[width=0.8\textwidth]{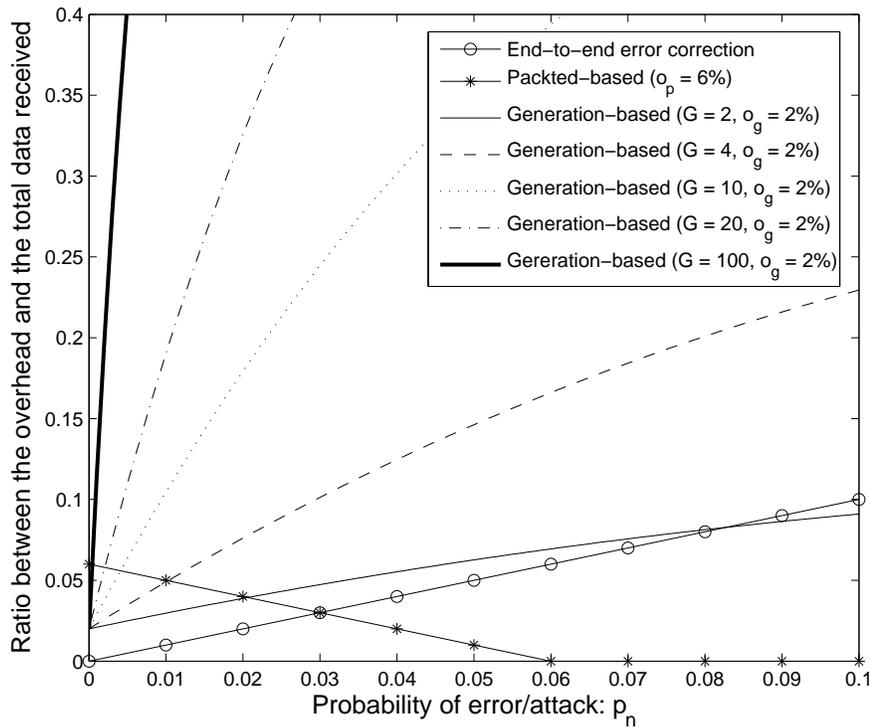}\end{center}\vspace*{-.5cm} \caption{Ratio between the expected overhead and the total data received by a node with $o_p = \frac{6}{100}(m+l)$, $o_g= \frac{2}{100}(m+l)G$ for $p_n \in [0, 0.1]$}
\label{fig:blownup}
\vspace*{-.6cm}
\end{figure}

%% file: conclusion.tex
In this paper, we studied the problem of Byzantine attacks in network coded P2P networks. We used randomly evolving graphs to characterize the impact of Byzantine attackers on the receiver's ability to recover a file. As shown by our analysis, even a small number of attackers can contaminate most of the flow to the receivers. Motivated by this result, we proposed a novel signature scheme for any network using RLNC. The scheme makes use of the linearity of the code, and it can be used to easily check the validity of all received packets. Using this scheme, we can prevent the intermediate nodes from spreading the contamination by allowing nodes to detect contaminated data, drop them, and therefore, only transmit valid data. We emphasize that there is no need of retransmission for the dropped data since the receivers can perform erasure correction, which is computationally cheaper than error correction.

We analyzed the cost and benefit of the signature scheme, and compared it with the end-to-end error correction scheme and the generation-based detection scheme. We showed that the overhead associated with our scheme is low. Furthermore, when the probability of Byzantine attack is high, it is the most bandwidth efficient. However, if the probability of attack is low, generation-based Byzantine detection schemes are more appropriate. 